\documentclass[12pt]{article}
\usepackage{graphicx}
\usepackage[round]{natbib}
\usepackage[colorlinks=true, linkcolor=blue, citecolor=blue,
urlcolor=black, pdfpagemode={UseNone}]{hyperref}
\usepackage{authblk} 
\usepackage{url}
\usepackage[font=small]{caption} 
\usepackage{amsmath,amssymb,bm,amsthm}
\usepackage{mathtools}
\usepackage{bbm}
\usepackage{color}
\usepackage{anyfontsize} 
\usepackage{setspace} 
\usepackage{silence} 
\graphicspath{{figures/}}

\newtheorem{proposition}{Proposition}

\theoremstyle{definition}

\newcommand{\di}{d} 
\newcommand{\p}{p} 
\newcommand{\penalt}{q} 

\DeclareMathOperator*{\argmin}{\mbox{argmin}}

\DeclareMathOperator{\tr}{\mbox{trace}}

\newcommand{\sspan}{\mbox{span}}
\newcommand{\ang}{\mbox{angle}}

\newcommand{\case}{\text{\scriptsize{case}}}
\newcommand{\cell}{\mbox{\scriptsize{cell}}}
\newcommand{\obs}{\mbox{\scriptsize{obs}}}
\newcommand{\new}{\mbox{\scriptsize{new}}}
\newcommand{\imp}{\mbox{\scriptsize{imp}}}

\newcommand{\eps}{\varepsilon}

\newcommand{\bzero}{\boldsymbol 0}
\newcommand{\bone}{\boldsymbol 1}
\newcommand{\bdot}{\boldsymbol{\cdot}}
\newcommand{\btop}{\boldsymbol{\top}}
\newcommand{\ba}{\boldsymbol a}
\newcommand{\bb}{\boldsymbol b}

\newcommand{\be}{\boldsymbol e}

\newcommand{\bv}{\boldsymbol v}
\newcommand{\bhx}{\boldsymbol{\widehat{x}}}

\newcommand{\bimpx}{\boldsymbol{x}^{\mbox{\scriptsize{imp}}}}
\newcommand{\bx}{\boldsymbol x}

\newcommand{\by}{\boldsymbol y}

\newcommand{\bA}{\boldsymbol A}
\newcommand{\bB}{\boldsymbol B}

\newcommand{\bD}{\boldsymbol D}
\newcommand{\bH}{\boldsymbol H}
\newcommand{\bI}{\boldsymbol I}

\newcommand{\bP}{\boldsymbol P}
\newcommand{\bhP}{\boldsymbol{\widehat{P}}}
\newcommand{\bR}{\boldsymbol R}

\newcommand{\bU}{\boldsymbol U}

\newcommand{\bV}{\boldsymbol V}

\newcommand{\bW}{\boldsymbol W}
\newcommand{\bX}{\boldsymbol X}
\newcommand{\bhX}{\boldsymbol{\widehat{X}}}

\newcommand{\bZ}{\boldsymbol Z}

\newcommand{\bbeta}{\boldsymbol \beta}
\newcommand{\bgamma}{\boldsymbol \gamma}
\newcommand{\bhgamma}{\boldsymbol{\widehat{\gamma}}}

\newcommand{\bTheta}{\boldsymbol \Theta}
\newcommand{\bhTheta}{\boldsymbol{\widehat{\Theta}}}

\newcommand{\bhbeta}{\boldsymbol{\widehat{\beta}}}

\newcommand{\bmu}{\boldsymbol \mu}
\newcommand{\hmu}{\widehat{\mu}}
\newcommand{\bhmu}{\boldsymbol{\widehat{\mu}}}
\newcommand{\bSigma}{\boldsymbol \Sigma}
\newcommand{\bhSigma}{\boldsymbol{\widehat{\Sigma}}}

\newcommand{\halpha}{\widehat{\alpha}}
\newcommand{\hbeta}{\widehat{\beta}}

\newcommand{\hsigma}{\widehat{\sigma}}

\newcommand{\hy}{\widehat{y}}

\newcommand{\sbr}[1]{\left[#1\right]}
\newcommand{\idx}{{p_1 \ldots p_L}}

\definecolor{blue}{RGB}{0,0,255}
\definecolor{red}{RGB}{255,0,0}

\begin{document}

\def\spacingset#1{\renewcommand{\baselinestretch}
{#1}\small\normalsize} \spacingset{1}

\title{\bf Cellwise Outliers}

\author[1]{Mia Hubert}
\author[2]{Jakob Raymaekers} 
\author[1]{Peter J. Rousseeuw}

\affil[1]{\small Section of Statistics and Data Science, 
  University of Leuven, Belgium, \hspace{10mm}
  mia.hubert@kuleuven.be, peter.rousseeuw@kuleuven.be}
\affil[2]{\small Department of Mathematics, University 
  of Antwerp, Belgium,  
  jakob.raymaekers@uantwerpen.be}

\date{August 2, 2026}

\maketitle

\bigskip
\begin{abstract}
In statistics and machine learning, the 
traditional meaning of the terms 
`outlier' and `anomaly' is a case in the
dataset that behaves differently from 
the bulk of the data, which raises suspicion 
that it may belong to a different 
population. But nowadays increasing 
attention is being paid to so-called
cellwise outliers, which are individual 
values somewhere in the data matrix
(or data tensor).
Depending on the dimension, even a 
relatively small proportion of outlying 
cells can pollute over half the cases, 
which is a problem for existing methods
that work by downweighting or dropping
outlying cases. It turns out that detecting
cellwise outliers as well as constructing
cellwise robust methods requires 
techniques that are quite different from 
the casewise setting. For instance, one 
has to let go of some intuitive 
equivariance properties. The problem is 
difficult, but the past decade has seen 
substantial progress. For high-dimensional 
data the cellwise approach is becoming 
dominant, and typically can deal with 
missing values as well. We review 
developments in the estimation of
location and covariance matrices as well
as regression methods, principal component
analysis, methods for tensor data, and 
various other settings.
\end{abstract}

\mbox{Tensor} data.

\spacingset{1.12} 

\newpage
\section{INTRODUCTION} 
\label{sec:intro}
Real-world data often contains elements
that do not follow the pattern 
suggested by the majority of the data. 
Such outliers may be gross errors with
the potential to severely distort a
statistical analysis, but they may also
be valuable pieces of information that 
warrant further inspection. But whatever
their cause, we need to be able to detect 
them.

The approach of robust statistics, 
as formalized by e.g. \cite{Huber1964} and 
\cite{hampel1986robust} has been to first
obtain a robust fit, i.e. a model that fits 
the majority of the data well, and then to
look for cases that deviate from that fit. 
An important assumption underlying this 
approach is that of casewise contamination,
in which a case is either an outlier or 
free of any contamination. 
(This is sometimes called rowwise 
contamination, because cases are often 
encoded as rows of the data matrix,
but the name casewise is more general.)
Different formalizations of casewise
contamination exist, but the most common 
one is to assume that the observed data 
was generated from 
a clean distribution $F$ with probability 
$1-\varepsilon > 0.5$ and from an arbitrary 
distribution $H$ with probability 
$\varepsilon < 0.5$. This is known as 
the {\it Tukey-Huber contamination model}.
A commonly used formulation is 
$X_\eps \sim F_{\varepsilon} \coloneqq 
 (1-\varepsilon)F + \varepsilon H$. 
The goal is to estimate the characteristics 
of $F$ given that we only observe 
$F_{\varepsilon}$ and, importantly, 
we don't assume anything about $H$.
The casewise contamination model can 
equivalently be written as
\begin{equation} \label{eq:cont_case}
   X_{\eps}=A^{\case} \odot X + 
   (\bone_{\di}-A^{\case}) \odot Z 
\end{equation}
where the Hadamard product $\odot$ multiplies 
vectors (and matrices) entry by entry, and 
$\bone_{\di}$ is a column vector with all 
$\di$ components equal to $1$. We assume that
$X$, $A^{\case}$ and $Z$ are mutually 
independent, with $X \sim F$ and $Z \sim H$. 
The $\di$-variate variable $A^{\case}$ has 
Bernoulli distributed marginals $A_j^{\case}$ 
for $j=1,\ldots,\di$ with success parameter 
$1-\eps^{\case}$, that are \textbf{fully 
dependent} in the sense that 
$P(A_1^{\case} = \ldots = A_{\di}^{\case})=1$.
\citet{alqallaf2009} call this the 
\textit{fully dependent contamination model}. 
It implies that on average, 
$(1-\eps^{\case})100\%$ of the cases are clean. 
The left panel of Figure~\ref{fig:casecellmixed} 
visualizes this setting for a toy data set with 
15 cases and 10 variables. Here 3 out of the 15 
cases (20\%) are casewise outliers.

\begin{figure}[!ht]
\centering
\vspace{2mm}
\includegraphics[width=1.0\textwidth]
   {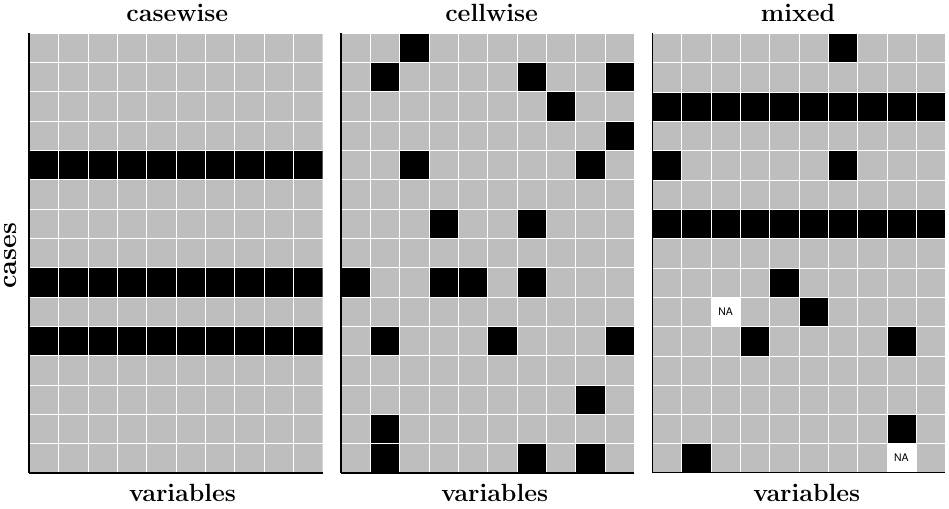}\\
\caption{Illustration of the casewise, cellwise, 
 and mixed contamination models. Black means
 outlying, and white cells are missing.}
\label{fig:casecellmixed}
\end{figure}

Casewise robust methods require that at least 
half of the cases are entirely uncontaminated. 
This assumption is often realistic in 
low-dimensional datasets, but less so 
in high dimensions because then there are many 
variables in which something can go wrong. 
Moreover, casewise robust methods work by 
downweighting or deleting all entries of the
outlying cases, whereas their outlying behavior 
might occur in only a few of them. This 
has motivated the study of {\it cellwise 
outliers}, which are deviating measurements 
(cells) that can occur anywhere in the data 
matrix.

The \textit{cellwise contamination model} 
assumes that the data are generated according to
\begin{equation} \label{eq:cont_cell}
  X_{\eps}=A^{\cell} \odot X + 
  (\bone_{\di}-A^{\cell}) \odot Z 
\end{equation}
with $X$ and $Z$ as in~\eqref{eq:cont_case}. 
The $\di$-variate variable $A^{\cell}$ has  
Bernoulli components $A_j^{\cell}$ with 
success probability $1-\eps^{\cell}_j$. 
If the components of $A^{\cell}$ are 
\textbf{independent}, we obtain the 
\textit{fully independent contamination 
model} of \citet{alqallaf2009}. 
Then on average each 
variable has $(1-\eps^{\cell}_j)100\%$ clean 
values, but even a relatively small proportion 
of outlying cells can contaminate over half 
the cases, which may cause casewise robust 
methods to fail. 
The middle panel of 
Figure~\ref{fig:casecellmixed} illustrates 
how a dataset with $22/150 \approx 15\%$ of 
outlying cells yields only 
$4/15 \approx 27\%$ entirely uncontaminated 
cases. It also illustrates that a case can 
have zero, one, a few, or many outlying 
cells. In general, the probability that a 
case contains at least one outlying cell is 
$1-(1-\eps^{\cell})^{\di}$ which grows quickly 
with the dimension $\di$. For instance, when
$\eps^{\cell}=0.05$ and $\di=14$ 
this probability is 51\%. It increases to 
97\% when $\di=70$, thus impacting the vast 
majority of cases. 

In real data, both types of outliers can 
occur simultaneously, and some values 
might be missing. The 
\textit{mixed contaminated and partially 
observed model} of \cite{cellPCA} says 
that
\begin{equation} \label{eq:cont_both}
  X_{\eps} = A \odot X + 
  (\bone_{\di}-A) \odot Z 
\end{equation}
where $A = A^{\case} \odot A^{\cell} \odot 
A^{\obs}$, the variable $A^{\case}$ is 
defined as in~\eqref{eq:cont_case}, and 
$A^{\cell}$ as in~\eqref{eq:cont_cell}. 
The entries $A_j^{\obs}$ in $A^{\obs}$ are 
binary variables with possible outcomes 1 
and NA, and 
$P(A_j^{\obs} = 1) = 1-\eps^{\obs}_j$. 
The right panel of 
Figure~\ref{fig:casecellmixed} shows a toy 
example. 

Different assumptions on the dependence 
structure of $X$, $Z$,  $A^{\case}$, 
$A^{\cell}$ and $A^{\obs}$ lead to different 
contamination models. For instance, when 
$A^{\obs}$ is independent of the other 
variables, the values in the dataset are 
missing completely at random. 

Long before the notion of cellwise 
outliers appeared in print in
\cite{alqallaf2009}, it was 
discussed informally by some. 
One of us remembers clearly that the 
topic was brought up in the early 
eighties in seminar talks by Alfio 
Marazzi and Werner Stahel at ETH 
Z\"urich. It was presented as an open
problem, with the formula
$1-(1-\eps^{\cell})^{\di}$. At that 
time the available tools appeared 
insufficient to address this 
challenge.

The remainder of the paper is organized 
as follows. Section~\ref{sec:detect}  
gives an overview of methods to detect
outlying cells. Section~\ref{sec:cov}
surveys the work on cellwise robust
estimation of multivariate location and
covariance matrices, and 
Section~\ref{sec:regr} reviews the work
on regression. Section~\ref{sec:highdim}
looks at methods specifically designed for 
high-dimensional data, and tensors are
treated in Section~\ref{sec:tensor}.
Section~\ref{sec:other} summarizes work
in other statistical settings, and
Section~\ref{sec:conc} concludes.

\section{DETECTING CELLWISE OUTLIERS}
\label{sec:detect}

The detection of cellwise outliers is 
not easy. A natural approach is to 
look at each individual variable, and
identify marginal outliers through, 
for example, robust z-scores. A more
refined version of this principle is 
the filter introduced by
\cite{gervini2002class}. When variables
are asymmetric, we can first apply the
robust transformation towards symmetry 
constructed by
\cite{raymaekers2021transforming}.

Detecting cellwise outliers by only
looking at marginal distributions has
some evident drawbacks. One can easily
imagine cellwise outliers that need not
be marginally outlying, especially 
in multivariate distributions with 
substantial correlations between the
variables. 
Consider for instance the TopGear dataset 
from the R package \texttt{robustHD}
\citep{robustHD}. It contains technical 
information of 297 cars such as their height, 
weight, acceleration, horsepower, and gas
mileage in miles per gallon (MPG).
The Ssangyong Rodius was reported with an 
acceleration time of 0 seconds, indicating 
that it achieves 60 mph instantly. This was 
obviously an error, and it can easily be 
identified from the marginal distribution 
as seen in the bottom left panel of 
Figure~\ref{fig:cellmap_Topgear}. On the 
other hand, the Fiat 500 Abarth appears to 
have a small width, but not marginally. The 
small width only becomes apparent when 
considering the car's other features, such 
as Top speed and Acceleration, as seen in 
the panels on the left. 

\begin{figure}[!ht]
\centering
\includegraphics[width=1.0\textwidth]
  {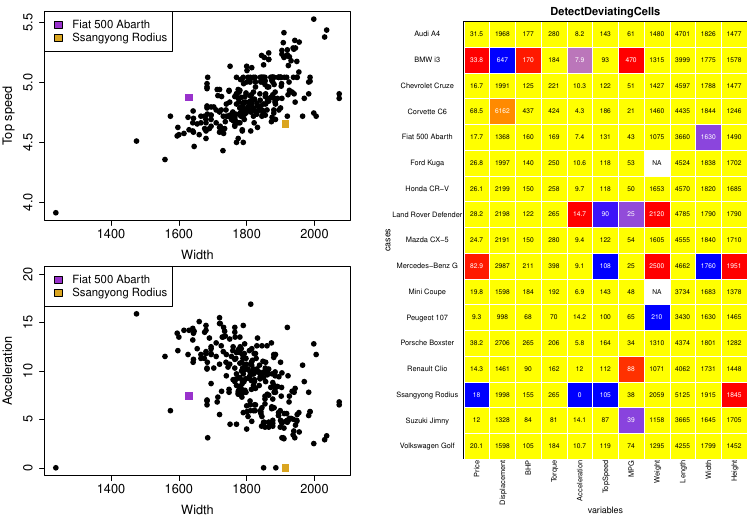}
\caption{Left: some variables of the Top
  Gear data, illustrating that cellwise 
  outliers need not be marginally outlying. 
  Right: DDC cellmap of these data.}
\label{fig:cellmap_Topgear}
\end{figure}

As a result, several proposals
have sought to incorporate multivariate
information into the detection of cellwise
outliers. For example, \cite{Leung2017}
extended the Gervini-Yohai filter to
a bivariate filter in the context of 
estimating location and scatter.
\cite{saraceno2021robust2} further extended
the approach to more than two dimensions. 

A popular approach for this purpose
is the Detect Deviating Cells (DDC) 
algorithm of \cite{Rousseeuw2018}
that is available in the \textsf{R}
package \texttt{cellWise}
\citep{cellWise}. For each variable, 
DDC fits robust simple linear regressions 
on the other variables in the multivariate
data set. By combining the predictions
from these regressions, DDC obtains
predictions for each variable.
The differences between these predictions
and the observed values then yield
cellwise residuals, based on which 
cellwise outliers can be identified.

The result can be visualized in a 
cellmap as shown in the right panel of 
Figure~\ref{fig:cellmap_Topgear}, 
displaying a part of the Top Gear 
data. The blue cells have a negative 
residual, that is, they have a lower 
value than expected based on the 
other variables. The red cells have a 
positive residual. We see that the DDC 
method does flag the width of the 
Fiat 500 Abarth, even though it is not 
marginally outlying. The DDC method has 
been applied in epidemiology
\citep{Viviani2023}, genetics
\citep{segaert2019robust},
geochemistry \citep{GRYGAR2024107416},
archaeology \citep{SANTOS2020102423, 
leggett2021migration},
and other fields. It also serves as a 
starting value for many iterative 
algorithms for cellwise robust estimation 
of more complex models.

These early proposals of detecting
cellwise outliers can and have been
used as preprocessing before 
applying traditional statistical 
methods. Unfortunately, this makes
the process of identifying cellwise 
outliers oblivious to the statistical
model that is being fit afterward.
This deviates 
from the established approach in 
traditional robust statistics of 
detecting outliers with respect 
to  a robust fit.
As we will see, recent advances 
in cellwise robust statistics are
moving in this direction, 
where the detection of outliers
is naturally paired with the model 
fitting. This typically
leads to more accurate detection
of harmful cellwise outliers, as
the influence of these outliers 
depends on the model being fit.

\section{MULTIVARIATE LOCATION AND\\ 
COVARIANCE MATRICES} \label{sec:cov}

Suppose we are given a data matrix $\bX$ with 
$n$ observations in $d$ dimensions. In this 
section we consider the estimation of a location 
vector and a covariance matrix when $d < n$. 
In casewise robust statistics, the traditional 
goal is to develop estimators for location and 
covariance that are affine equivariant. More 
precisely, consider any affine transformation 
$\bx \mapsto \bv + \bA\bx$ where $\bx$ and 
$\bv$ are in $\mathbb{R}^{d \times 1}$ and 
$\bA\in \mathbb{R}^{d\times d}$ is nonsingular. 
The location estimator $\bhmu$ and covariance 
estimator $\bhSigma$ are called affine 
equivariant when
\begin{equation} \label{eq:affloccov}
\begin{aligned}
 \bhmu(\bone_n\bv^{\btop}+\bX\bA^{\btop} )
  =&\; \bv + \bA\,\bhmu(\bX) \\
 \bhSigma(\bone_n\bv^{\btop}+\bX\bA^{\btop} )
  = &\; \bA\, \bhSigma(\bX) \bA^{\btop} , 
\end{aligned}
\end{equation}
where $\bone_n$ is the $n \times 1$
vector of ones. 
These equivariances say that the estimators 
behave naturally under affine transformations. 
Under this equivariance condition it turned
out one can construct highly casewise robust 
estimators of location and covariance, in 
the sense that almost 50\% of the cases can
be replaced before the estimator becomes 
completely uninformative.

In this section we consider the more general
contamination model of~\eqref{eq:cont_both}.
We assume that the clean data $\bX$ are 
sampled independently from the model $F$, 
say a multivariate normal distribution with mean 
$\bmu$ and covariance $\bSigma$. Afterward
the data can be contaminated by casewise and
cellwise outliers. Despite affine equivariance
being natural, we will see that it cannot be 
reconciled with robustness under cellwise 
contamination. This indicates that robustness 
to cellwise outliers requires a paradigm shift 
away from the casewise framework.

\subsection{Breakdown of location and 
            covariance estimators}
            
A well-known measure of the {\it casewise} 
robustness of a location estimator $\bhmu$ 
is its finite-sample breakdown value 
\citep{donoho1983}. 
Informally, the casewise breakdown value of 
$\bhmu$ at the dataset $\bX$ is the smallest
fraction of cases that needs to be replaced 
to carry the resulting estimate arbitrarily 
far from $\bhmu(\bX)$, the estimate on 
the original data. More precisely, denote by 
$\bX^m$ any corrupted data matrix obtained by 
replacing at most $m$ cases in $\bX$ by 
arbitrary points. The casewise breakdown 
value of $\bhmu$ at $\bX$ is defined as 
\begin{equation} \label{eq:bdvloc}
  \varepsilon^*_n(\bhmu, \bX)=
  \min \left\{\frac{m}{n}:\;
	\sup_{\bX^m}{\left|\left|\bhmu(\bX^m) - 
	\bhmu(\bX)\right|\right|} = 
	\infty\right\}.
\end{equation}

For the \textit{cellwise} breakdown value we
instead replace cells. In this setting $\bX^m$
becomes any contaminated data matrix obtained by 
replacing at most $m$ cells in each column of 
$\bX$ by arbitrary numbers. The 
{\it finite-sample cellwise breakdown value} of 
$\bhmu$ at $\bX$ is then defined as 
in~\eqref{eq:bdvloc}, but with this new $\bX^m$.
The old $\bX^m$ is a special case of the new
version, obtained by concentrating all the 
contaminated cells in a fraction of the rows,
so the casewise breakdown value of an
estimator is an upper bound on its cellwise 
breakdown value. 
\cite{lopuhaa1991breakdown} showed that the
casewise breakdown value of any translation 
equivariant location estimator is at most
$\lfloor (n+1)/2 \rfloor/n$, so this is also
an upper bound on the cellwise breakdown 
value. Note that translation equivariance is
the special case of the affine equivariance
in~\eqref{eq:affloccov} for $\bA = \bI$, 
the $d \times d$ identity matrix.

However, we cannot have affine equivariance
for arbitrary nonsingular $\bA$, and not even
orthogonal equivariance, in which $\bA$ is
any orthogonal matrix. That is because an 
orthogonal transformation of the data 
can change all the cells. The formal
result says

\begin{proposition} \label{prop:loc}
The cellwise breakdown value of any 
orthogonally equivariant location estimator 
$\bhmu$ at any dataset $\bX$ is bounded by
\begin{equation} \label{eq:locbdv}
  \varepsilon^*_n(\bhmu, \bX) \leqslant
  \left\lceil \frac{n}{d} 
  \right\rceil/n\, \approx \frac{1}{d}\;.
\end{equation}
\end{proposition}

More information on this result can be found 
in the Appendix. Intuitively, all 
the data can be moved to a far-away hyperplane 
by changing a single cell in each row of $\bX$,
and then $\bhmu(\bX^m)$ has to lie on that
hyperplane.

We conclude that orthogonal equivariance
cannot be reconciled with cellwise
robustness. For the estimation of location,
there is a natural solution in using 
a univariate robust location estimate for each 
variable separately, and combining them into a 
point. Such an estimator is not orthogonally 
equivariant, and inherits the breakdown value 
of the univariate location estimator.

Unfortunately, for estimating a covariance 
matrix there is no coordinatewise trick. Note 
that a covariance matrix estimator can break 
down in two ways, namely through having an 
arbitrarily large or an arbitrarily small 
eigenvalue. More precisely, the 
{\it cellwise explosion 
breakdown value} of the covariance 
estimator $\bhSigma$ is defined as
\begin{equation} \label{eq:explosion}
  \varepsilon^+_n(\bhSigma, \bX)=
  \min \left\{\frac{m}{n}:\;
  \sup_{\bX^m}\lambda_1(\bhSigma
  (\bX^m)) = +\infty\right\}
\end{equation}
where $\lambda_1$ denotes the largest 
eigenvalue. The {\it cellwise 
implosion breakdown value} of 
$\bhSigma$ is defined as
\begin{equation} \label{eq:implosion}
  \varepsilon^-_n(\bhSigma, \bX) =
  \min \left\{\frac{m}{n}:\;
  \inf_{\bX^m}\lambda_d(\bhSigma
  (\bX^m)) = 0 \right\}
\end{equation}
where $\lambda_d$ is the smallest 
eigenvalue. 
Implosion implies that the estimated
covariance matrix is singular.
A singular covariance matrix is often not 
very useful, as it is not invertible and
therefore cannot be used to perform
standard statistical analyses such as
the computation of Mahalanobis distances 
or discriminant analysis.

In the casewise setting, it is 
known that affine equivariant
covariance estimators can have a
breakdown value of at most
$\lfloor (n-d+1)/2\rfloor /n 
\approx 0.5$\,. 
It turns out that obtaining a high
cellwise breakdown value for the 
covariance matrix is impossible under 
affine equivariance.

\begin{proposition} \label{prop:impl}
The cellwise implosion breakdown value
of any affine equivariant covariance 
estimator $\bhSigma$ is bounded for any 
dataset $\bX$ by
\begin{equation} \label{eq:impl}
  \varepsilon^-_n(\bhSigma, \bX) 
  \leqslant \left\lceil \frac{n-1}{d} 
  \right\rceil/n\, \approx \frac{1}{d}\;.
\end{equation}
\end{proposition}
Since most common covariance estimators
are affine equivariant, they have a low
cellwise breakdown value of approximately
$1/d$. An interesting example is the 
classical empirical covariance matrix, 
that has a high {\it casewise} implosion 
breakdown value of $(n-d)/n \approx 1$, 
but this low {\it cellwise} breakdown value.

While the above results seem discouraging, 
we will see in the next section
that it is still possible to construct 
estimators of covariance with a high 
cellwise breakdown value.

\subsection{Cellwise robust estimators of 
location and covariance} 
\label{sec:cov_methods}

While location can be estimated 
coordinatewise, this strategy is no 
longer viable for the estimation of 
covariance matrices. \cite{VanAelst2011}
proposed the first cellwise robust 
covariance estimator. Their cellwise
Stahel-Donoho estimator is inspired
by the classical Stahel-Donoho estimator
for casewise robust covariance estimation 
\citep{stahel1981robuste,
donoho1982breakdown}. The estimator was 
later surpassed in performance by the 
two-step generalized S-estimator 
(TSGS) of \citep{Agostinelli2015, 
Leung2017}. TSGS first filters the data 
for cellwise outliers, and puts 
the flagged cells to missing. In a 
second step, it applies the generalized 
S-estimator (GSE) of \cite{danilov2012}
for casewise robust covariance 
estimation with missing data. This 
combination of filtering and estimation 
was the topic of an extensive comparative
study by \cite{danilov2010robust}.
The conclusion was that some alternative
approaches were computationally 
burdensome. TSGS therefore incorporates
a bivariate filter to keep 
the computation time manageable. 
TSGS and GSE are implemented in 
the \textsf{R}-package \texttt{GSE}  
\citep{Leung2019}.

TSGS was the state-of-the-art method for 
estimating covariance matrices under cellwise 
contamination for quite some time. More 
recently, \cite{cellMCD} proposed the cellMCD
method that combines the identification of 
suspicious cells and the estimation of the 
covariance matrix, based on earlier research 
in \cite{cellHandler}. 

CellMCD is a cellwise robust extension of the 
casewise robust Minimum Covariance Determinant 
estimator of \cite{rousseeuw1984least}, that
maximizes the likelihood of a subset of 
{\it cases}. It is a combinatorial approach 
with the aim of identifying a subset of cases 
that tightly fit together. In the cellwise
paradigm, we instead consider a subset of 
trusted {\it cells}. The identity of the cells
used is stored in an $n \times d$ binary 
matrix $\bW$. Here $\bW_{ij} = 1$ means we
use the cell $x_{ij}$ (it is considered clean), 
and $\bW_{ij} = 0$ means we do not use cell
$x_{ij}$ (it is considered contaminated).

A natural measure of how well the used cells 
fit together is the observed likelihood
$L(\bX, \bW, \bmu, \bSigma)$. If we knew $\bW$
ahead of time, we would be in the usual
missing data setting, with $\bW_{ij} = 0$
indicating missing values. Then we could simply
maximize $L(\bX, \bW, \bmu, \bSigma)$ by
the EM algorithm, yielding $(\bhmu,\bhSigma)$.

However, here $\bW$ is not known in advance,
so it needs to be estimated as well. We cannot
just select $(\bW,\bhmu,\bhSigma)$ by 
maximizing $L(\bX, \bW, \bmu, \bSigma)$ over
all possibilities, because the resulting
$\bW$ could take out too many cells. That is 
why cellMCD instead minimizes
\begin{equation} \label{eq:cellMCD}
 -2\log L(\bX, \bW, \bmu, \bSigma) + 
 \sum_{j=1}^d \penalt_j (\# 
 \mbox{ flagged cells in variable }j)
\end{equation}
with respect to $\bmu$, $\bSigma$, and $\bW$.
The second term penalizes flagging many 
cells (setting their $\bW_{ij}$ to zero), 
which helps the efficiency. The constants 
$\penalt_j$ are fixed before running the 
algorithm, and determine the quantile at 
which cells are considered outlying,
typically $99\%$.
The minimization is carried out under two 
constraints. The first is a lower bound on
the smallest eigenvalue of $\bSigma$, needed
for $L(\bX, \bW, \bmu, \bSigma)$ to exist.
The second says the number of unflagged cells
in each column must be at least $h$, a
number given by the user that must be at 
least $n/2$ and is $0.75n$ by default.

The objective~\eqref{eq:cellMCD}
is minimized as follows. Starting from a 
robust initial guess for $(\bmu,\bSigma)$, 
the algorithm iterates between updating 
$\bW$ (i.e., flagging suspicious cells),
and re-estimating $(\bmu, \bSigma)$. 
Each step is guaranteed to lower the 
objective function, so we know the 
iterations will converge. CellMCD has a 
high cellwise  breakdown value that is at 
most $(n - \lfloor n/2\rfloor)/n$, 
depending on the choice of $h$.

CellMCD implicitly uses a hard-thresholding
rule. If a cell is flagged as outlying,
all information in that cell is discarded.
In some settings, it might be more 
appropriate to give smoother weights
to the cells of a data matrix, and obtain 
estimates based on these weights. To this
end \cite{cwMLE} proposed a method to 
statistically analyze data with cellwise 
weights that can take any value between
$0$ and $1$. It provides a canonical 
definition of cellwise weighted likelihood, 
that can be used to estimate location and 
covariance for given weights.

\section{REGRESSION} 
\label{sec:regr}

We denote the linear regression model as
\begin{equation} \label{eq:reg}
 \by = \alpha + \bX \bbeta + \mbox{noise} 
\end{equation}
where $\by = (y_1,\ldots,y_n)$ is the 
$n \times 1$ column vector of responses,
and the $n \times \p$ matrix $\bX$ of 
regressors has the rows $\bx_i^{\btop}$. 
We want to estimate the $(\p+1) \times 1$ 
parameter vector $\bgamma = (\alpha,\beta_1,
\ldots, \beta_{\p})^{\btop}$ that combines
the intercept and the slopes. 

In casewise robust linear regression, 
an observed pair $(\bx_i, y_i)$ is 
considered to be either outlying or 
entirely clean. 
For cellwise robust regression the setup 
is different.
It remains true that the response $y_i$
is either outlying or clean because it 
is univariate, so it has only one cell.
The main difference lies in the 
contamination of the regressor variables. 
Instead of working with either outlying
or entirely clean 
$\bx_i^{\btop} = (x_{i1},\ldots,x_{i\p})$, 
we now allow for some of the cells 
$x_{ij}$ to be contaminated while the 
others are clean. 
We would like to use the clean cells of
$\bx_i$ in the estimation of $\alpha$ 
and $\bbeta$, but at the same time limit 
the harmful influence of the bad cells
of $\bx_i$\,.

\subsection{Breakdown in regression}
\label{reg_bdv}

The cellwise breakdown value of a
regression estimator is defined as 
in~\eqref{eq:bdvloc}, by replacing the
vector $\bhmu$ by the vector 
$\bhgamma = (\halpha,\hbeta_1,
\ldots, \hbeta_{\p})^{\btop}$ and
replacing $\bX$ by $(\bX,\by)$.

Existing regression methods are
typically {\it regression equivariant},
that is, for all $(\p+1)\times 1$ 
vectors	$(\widetilde{\alpha},
\widetilde{\beta}_1, \ldots, 
\widetilde{\beta}_{\p})^{\btop}$
it holds that
\begin{equation}\label{eq:regeq}
 \bhgamma(\bX,\by +
 \widetilde{\alpha} +
 \bX(\widetilde{\beta}_1, \ldots, 
\widetilde{\beta}_{\p})^{\!\btop}) =
 \bhgamma(\bX,\by)+ 
 (\widetilde{\alpha},
\widetilde{\beta}_1, \ldots, 
\widetilde{\beta}_{\p})^{\!\btop}
\end{equation}
for all datasets $(\bX,\by)$. This
is similar to translation equivariance 
for location estimators.
Another common property is {\it scale
equivariance}, meaning that for any
constant $c$ we have
\begin{equation}\label{eq:scaleq}
  \bhgamma(\bX,c\,\by)=
			  c\,\bhgamma(\bX,\by)\;.
\end{equation}
However, these properties preclude
cellwise robustness.

\begin{proposition} \label{prop:reg}
The cellwise breakdown value of any 
regression and scale equivariant 
estimator $\bhgamma$
is at most
\begin{equation} \label{eq:bdreg}
 \left\lceil \frac{n-1}{\p+1} 
 \right\rceil/n\, \approx 
 \frac{1}{\p+1}\;.
\end{equation}
\end{proposition}

So if we want a regression estimator with 
a better cellwise breakdown value, we have 
to let go of the combination of equivariance 
properties~\eqref{eq:regeq}
and~\eqref{eq:scaleq}.

\subsection{Cellwise robust regression
methods} \label{sec:reg_methods}

One of the earliest proposals for 
cellwise robust regression was by
\cite{leung2016robust}. Their model 
assumes joint multivariate normality
of the pairs $(\bx_i, y_i)$ before
the cells are contaminated.
They start by estimating the location 
$\bhmu$ and covariance matrix $\bhSigma$ 
of the joint distribution of the
pairs $(\bx_i, y_i)$ by the TSGS 
estimator of \cite{Agostinelli2015},
with a modified filter in the first step. 
Then the slope coefficients and the
intercept are estimated from $\bhmu$
and $\bhSigma$ by the usual formulas
\begin{equation} \label{eq:beta}
 \bhbeta := \bhSigma_{xx}^{-1}
 \bhSigma_{xy}\;\;\;
 \mbox{ followed by }\;\;\;
 \halpha = \hmu_y - 
 \bhmu_x^{{\btop}}\bhbeta \;.
\end{equation}
The robustness of this estimator stems
from the TSGS method, but this approach 
leans heavily on the assumption of joint 
normality of the clean data.

An alternative proposal is the 
shooting S-estimator 
of \cite{ollerer2016shooting}. 
It uses the coordinate descent 
algorithm \citep{bezdek1987}, also 
called ``shooting algorithm''. 
A typical step is to fit only the 
$j$-th slope coefficient $\beta_j$ 
while keeping the other slopes fixed 
at their previous values:
\begin{equation} \label{eq:shooting}
  y_i^{(j)} \sim \;
  \halpha_j^{\mbox{\tiny{new}}} + 
  \hbeta_j^{\mbox{\tiny{new}}} x_{ij} 
  \;\;\;\;\;\mbox{ where }\;\;\;\;\;
  y_i^{(j)} \coloneqq y_i - 
  \bx_{i,-j}^{\btop} 
  \bhbeta_{-j}^{\mbox{\tiny{old}}} 
  \;\;.
\end{equation}
This can be seen as a simple regression 
of a new response $\by^{(j)}$ on the 
$j$-th predictor. 
The shooting S-estimator makes two 
changes to the coordinate descent
algorithm. The first is that in the 
simple regression 
of~\eqref{eq:shooting} the OLS 
regression is replaced by a casewise
robust S-estimator. 
The second is that the response 
$\by^{(j)}$ is modified in each 
iteration, to avoid that the cellwise 
outliers propagate to the new 
responses in the iteration steps. 
\cite{bottmer2022sparse} recently 
proposed a version of shooting S
regression for high dimensions that
uses hard thresholding, yielding 
sparse estimates of the regression 
coefficients.

The variable selection method of 
\cite{su2024robust} first computes a 
robust covariance matrix of the joint
distribution of the pairs $(\bx_i, y_i)$.
For this they use a method based on
pairwise correlations, using those of  
\cite{gnanadesikan1972robust} or the  
Gaussian rank correlation studied by 
\cite{boudt2012gaussian}.
They then plug the square root of this 
covariance matrix into the adaptive 
lasso objective.
\cite{toka2021robust} propose a 
three-step procedure consisting of (i) 
identifying marginal cellwise outliers 
using robust z-scores and setting them
to missing, (ii) running the casewise 
robust MCD estimates of location and
covariance on the complete cases followed
by imputing the NA's as in the E-step 
of the EM algorithm, and (iii) applying 
a robust lasso regression to the 
imputed data. 
\cite{saraceno2021robust} carry out
robust seemingly unrelated regressions 
using several univariate MM-regressions, 
and apply the TSGS method of 
\cite{Agostinelli2015} to the residual
matrix. Another proposal is the algorithm 
by \cite{CRlasso}. It optimizes an 
objective function inspired by 
mean-shift penalized objectives for 
casewise robust regression 
\citep{she2011outlier}. Recently
\cite{Christidis2026} constructed a 
cellwise robust ensemble regression 
method.

The cellwise least trimmed squares
(cellLTS) method proposed by 
\cite{cellLTS} works in a different way.
Its first step is to apply the cellMCD
method of Section~\ref{sec:cov_methods}
to the regressor matrix $\bX$. 
Next, cellwise 
outliers as well as NA's are imputed, 
yielding a cleaned matrix $\bhX$ of
regressors. Finally, the casewise 
least trimmed squares (LTS) method is
applied to regress $\by$ on $\bhX$.
To avoid collinearity issues, the LTS
objective now includes a penalty term
$\lambda||\bbeta||_2^2$ with a small
$\lambda$, so it is a ridge version of 
LTS, with a tailor made algorithm. 
Moreover, cellLTS has an option to
symmetrize the data, for the situation
that some variables are rather skewed.
This is applied to $[\bX\;;\;\by]$, and 
then the location $\bmu$ of $\bX$ and 
the intercept $\alpha$ of the regression 
must be estimated afterward.

\begin{figure}[!ht]
\centering
\includegraphics[width=1.0\columnwidth]
  {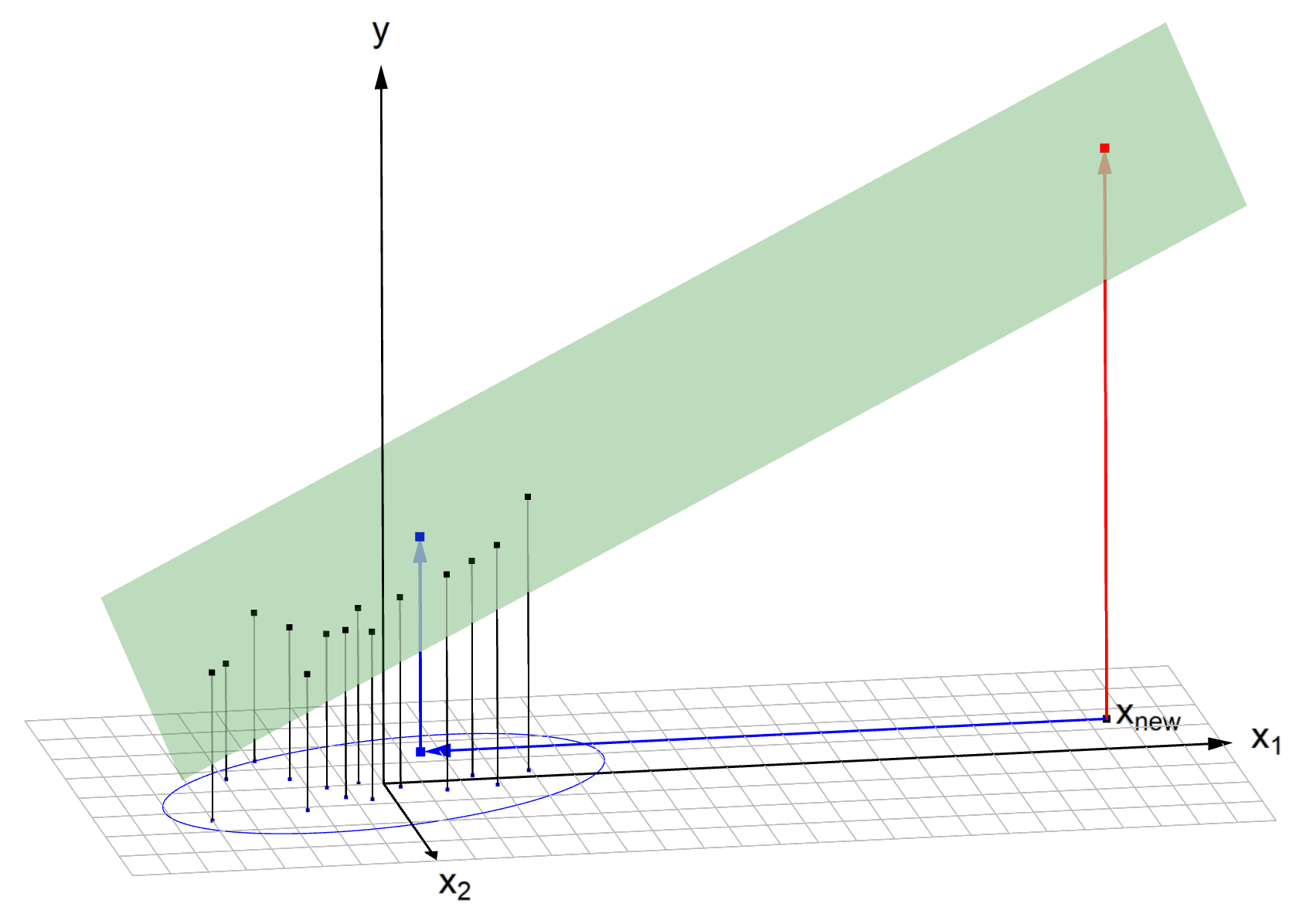}
\caption{Illustration of out-of-sample
   prediction with the cellLTS method.}
\label{fig:cellLTS}
\end{figure}

Another novelty of cellLTS is that it
is geared to making robust 
out-of-sample predictions. Indeed, 
when a new regressor vector $\bx_{\new}$
arrives, it may have outlying cells. 
Therefore we should not just use the 
standard formula $\hy = \halpha + 
\bx_{\new}^{\btop}\bhbeta$.
This is illustrated in 
Figure~\ref{fig:cellLTS}, where $\by$ is
regressed on two variables. The black
points are the in-sample (training)
data. The majority of the bivariate 
regressor vectors $(x_{i1},x_{i2})$ lie 
inside the blue ellipse.
At the bottom right is the new point
$\bx_{\new}$, whose first cell is a
distant cellwise outlier. Simply
multiplying that cell by $\hbeta_1$
gives it a very big effect: the 
predicted value would be obtained by 
following the red arrow upward until
it meets the regression plane.
Instead, the cellLTS prediction starts 
by imputing the first cell of
$\bx_{\new}$ along the blue
horizontal arrow. For the prediction
we then go upward on the blue
vertical arrow, yielding a more
natural prediction.

The in-sample predictions by cellLTS
use the exact same mechanism, that 
also works for missing cells. It was
verified by simulation that the cellLTS
predictions outperformed those of 
earlier methods, and it was illustrated 
on real data with outlying cells.

CellLTS has the advantage that its
cellwise breakdown value can be
derived. CellLTS
uses a value $h$ like the cellMCD
and casewise LTS methods.

\begin{proposition} \label{prop:cellLTS}
If $\bX$ is in general position and 
$h \geqslant [(n+p+1)/2]$, then the
cellwise breakdown value of cellLTS is
$\eps_n^*(\bhgamma, \bZ) = (n-h+1)/n$.
This is also its casewise breakdown
value.
\end{proposition}

\section{HIGH-DIMENSIONAL DATA}
\label{sec:highdim}

\subsection{Detecting outlying cells in
   high dimensions}\label{sec:FastDDC}

The methods discussed so far were primarily 
designed to deal with low-dimensional data. 
Most of them require $d<n$ and have a strongly 
deteriorating performance as the dimension 
increases. However, the problem of cellwise 
outliers becomes more pronounced and thus 
increasingly relevant as the dimension of the 
data increases. The high-dimensional setting 
is thus a natural one for cellwise robustness, 
and several proposals for detecting outliers, 
dimension reduction and covariance estimation 
have been made. While working componentwise is 
usually very fast, it has the obvious limitation 
that the outlyingness of a cell is determined 
only by considering its marginal distribution. 
This has motivated the development of many other 
techniques, that have been shown to perform 
better in challenging scenarios.

For detecting cellwise outliers, 
\cite{raymaekers2021fast} introduced FastDDC, 
a computationally efficient version of the DDC 
method of Section~\ref{sec:detect}. FastDDC 
relies on a fast and robust screening technique 
for detecting linearly dependent variables, 
that are then used in an ensemble of simple 
linear regressions to detect cellwise outliers. 
For this purpose the variables are first 
transformed with componentwise transformations 
that limit the influence of outliers. Next,  
these transformed variables are subjected 
to a fast nearest-neighbor algorithm in dual 
space.

In order to obtain a cellwise robust 
correlation matrix in very high dimensions
the so-called {\it wrapping transformation}
of~\cite{raymaekers2021fast} can be used.
They gave an illustration where a dataset
with 921,600 dimensions was analyzed by
wrapping followed by a truncated 
singular value decomposition.

\subsection{Principal component analysis} 
\label{sec:PCA}

Principal component analysis (PCA) is an 
essential tool of multivariate statistics. 
Several formulations of classical PCA
exist, e.g. based on the spectral 
decomposition of the covariance matrix
or the singular value decomposition of
the centered data matrix.
One way of looking at it, is that PCA
approximates the dataset $\bX$ 
with $d$ columns by a new dataset 
$\bhX$ that lies in an affine subspace of 
dimension $k < d$ such that the Frobenius 
norm $||\bX - \bhX||_F$ is minimized. 
Several proposals exist for casewise 
robust PCA, including spherical
PCA \citep{locantore1999robust}, PCA
based on least trimmed squares orthogonal
regression \citep{maronna2005principal}, 
the PCAgrid projection pursuit method 
\citep{croux2007}, and the hybrid 
ROBPCA method \citep{hubert2005robpca}. 
There are also approaches for casewise 
robust sparse PCA, see 
\cite{croux2013sparsePCA}, 
\cite{HubertROSPCA},
and \cite{wang2020sparse}.

For cellwise outliers, things are more 
difficult. One complication stems from 
the fact that classical PCA projects the 
data orthogonally on a lower dimensional
subspace. Projecting a case with 
cellwise outliers is problematic, as the 
outlying cells can propagate over all
cells. Therefore, whenever the algorithm
makes a projection, the outlying cells
should be addressed.

\subsection{Breakdown of PCA}
\label{sec:PCA_bdv}

Following \cite{cellPCA} we formulate a
PCA method as an estimator of the 
\mbox{$k$-dimensional} principal subspace 
$\bhP(\bX)$ of $\bX$. This definition is
quite general and does not require to specify
principal directions inside the subspace.

For the breakdown value, we need an appropriate 
definition of breakdown in the context of PCA. 
We define the finite-sample cellwise 
breakdown value of $\bhP$ at the dataset $\bX$ 
as the smallest fraction of contamination needed 
to get a maxangle (defined below) of $\pi/2$ 
between the resulting affine subspace and the 
original one:
\begin{equation} \label{eq:bdvPCA}
  \varepsilon^*_n(\bhP, \bX)=
  \min \left\{\frac{m}{n}:\;
  \sup_{\bX^m} \; \mbox{maxangle}\Big(
  \bhP(\bX^m), \bhP(\bX)\Big)
  = \frac{\pi}{2} \right\}.
\end{equation}
Here the maxangle between subspaces $\bA$ and
$\bB$ is defined as follows. Shift the affine 
subspaces so they pass through the origin, 
making them linear subspaces $\bA_0$ and
$\bB_0$. Then
\begin{equation*}
  \mbox{maxangle}(\bA, \bB) :=
  \max_{\ba \in \bA_0,\,||\ba||=1}\Big(
  \min_{\bb \in \bB_0,\,||\bb||=1}
  \ang(\ba,\bb)\Big).
\end{equation*}
The maxangle equals $\pi/2$ iff there is a 
vector in $\bA_0$ that is orthogonal to all 
the vectors in $\bB_0$ (and vice versa).
It is provided by the \textsf{R} 
function \texttt{pracma::subspace()}.

In casewise robustness, it is natural to
assume that a PCA method is orthogonally
equivariant. This means that when you
apply an orthogonal transformation $\bR$ 
to your data, the principal subspace will
transform accordingly:
$\bhP(\bX \bR^{\top}) = \bR\,\bhP(\bX)$.
But also here, this natural condition
would prevent a high cellwise breakdown
value:

\begin{proposition} \label{prop:PCA}
The cellwise breakdown value of any
orthogonally equivariant PCA estimator 
$\bhP$ is bounded by
\begin{equation*} \label{eq:pcabdv}
  \varepsilon_n^*(\bhP, \bX) \leqslant
  \left\lceil \frac{n-1}{d} 
  \right\rceil/n\, \approx \frac{1}{d}\;.
\end{equation*}
\end{proposition}

\subsection{Cellwise robust PCA
methods} \label{sec:PCA_methods}

An early proposal for the related 
problem of cellwise robust factor
analysis was made by \cite{croux2003RAR}.
The first cellwise robust PCA method
was proposed by \cite{de2003framework} 
and improved by \cite{maronna2008robust}. 
Instead of minimizing the Frobenius norm 
of $||\bX - \bhX||_F$ they minimize a 
robust approximation error by applying 
a bounded function $\rho$ to each cell 
of $\bX - \bhX$ and summing the 
resulting errors. The minimization 
is carried out by iteratively solving 
weighted least squares problems.
The bounded $\rho$ function guarantees
that a single cell 
cannot dominate the approximation error.
The MacroPCA method proposed by
\cite{hubert2019macropca} was the first 
all-in-one method capable of handling
missing values as well as being robust
to cellwise and casewise outliers. 
It starts from DDC and 
incorporates elements of ROBPCA. 
It imputes
cells that caused data points to lie 
far from the fitted subspace, and
produces cellwise residuals that can
be plotted.

The recent cellPCA method of \cite{cellPCA} 
improves on MacroPCA by using a single 
objective function that combines two robust 
loss functions. It approximates
$\bX$ by $\bhX := \widehat{\bX^0} 
+ \bone_n\bhmu^{\btop}$ with 
$\widehat{\bX^0}$  an $n \times d$ matrix 
of rank $k < d$, obtained by minimizing 
\begin{equation} \label{eq:cellPCA}
   \frac{\hsigma_2^2}{m}\sum_{i=1}^{n} m_i \rho_2\! 
  \left(\frac{1}{\hsigma_2}\sqrt{\frac{1}{m_i}
  \sum_{j=1}^{\di} m_{ij}\, \hsigma_{1,j}^2\,
  \rho_1\!\left(\frac{x_{ij}-\widehat{x}_{ij}}
  {\hsigma_{1,j}}\right)}\, \right)\, .
\end{equation}
Here $m_{ij}$ is 0 if $x_{ij}$ is missing and 1 
otherwise, $m_i=\sum_{j=1}^{d} m_{ij}$\,,
and $m=\sum_{i=1}^{n} m_i$\,.
The idea is to limit the impact of  
the standardized cellwise residuals $(x_{ij} - 
\widehat{x}_{ij})/\hsigma_{1,j}$ by a bounded loss 
function $\rho_1$. Similarly, $\rho_2$ mitigates 
the effect of casewise outliers.

The objective function is minimized by an 
iteratively reweighted least squares algorithm. 
For each case it yields a casewise weight 
$w^{\case}_{i}$ and cellwise weights 
$w^{\cell}_{ij}$ that transition smoothly from 
1 for clear inliers to 0 for far outliers.

The cellPCA method also provides imputed cases
$\bimpx_i$  in which suspicious cells are cleaned
and missing cells are filled in. The modified 
cells are defined such that $\bimpx_i$ is 
shrunk toward the PCA subspace, and that the 
orthogonal projection of $\bimpx_i$ coincides 
with the fitted $\bhx_i$\,. This imputation is 
illustrated in the top panel of 
Figure~\ref{fig:cellPCA} for a 1-dimensional 
principal subspace in 3-dimensional
space. Most points are projected orthogonally on 
the fitted subspace, but the first cell of point 
$a$ is imputed before projecting it, and the 
second cell of $b$, and the third cell of $c$.
None of these cells are marginally outlying.

\begin{figure}[ht]
\centering
\includegraphics[width=0.9\textwidth]
  {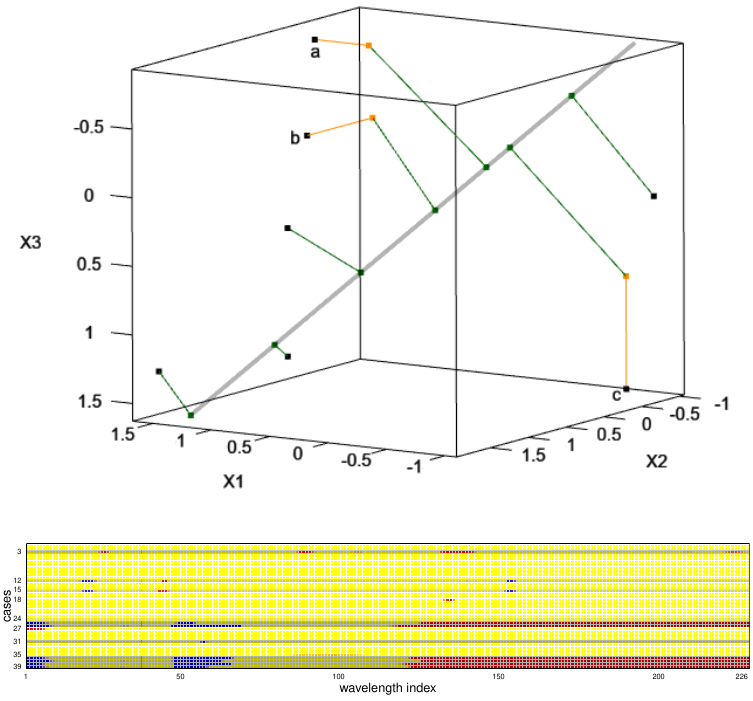}
\caption{Upper panel: Illustration of cellPCA
imputation for $d=3$ and $k=1$. The first cell 
of point $a$ was imputed before projecting it on 
the principal subspace, the second cell of $b$, 
and the third cell of $c$. Lower panel: Shaded 
residual cellmap of the octane data.}
\label{fig:cellPCA}
\end{figure}

After obtaining the PCA subspace, the algorithm 
carries out an additional step by applying the 
casewise MCD estimator to the raw $k$-variate 
scores. This yields an orthogonal set of $k$ 
loading vectors collected in the loadings matrix
$\bV \in \mathbb{R}^{d \times k}$, as well as a 
scores matrix  $\bU \in \mathbb{R}^{n \times k}$ 
such that 
$\bhX := \widehat{\bU} \widehat{\bV}^{\btop} + 
 \bone_n\bhmu^{\btop}$.

Based on the cellPCA results, a \textit{shaded 
residual cellmap} can be drawn 
\citep{Hirari:cellGraphics}. The lower panel of
Figure~\ref{fig:cellPCA} is the map of the 
octane data of \cite{Esbensen:MultiAnalysis} that 
contains near infrared (NIR) absorbance spectra of 
$n = 39$ gasoline samples over $d = 226 > n$ 
wavelengths. First the cells are colored as in 
Figure~\ref{fig:cellmap_Topgear} according to their 
standardized cellwise residual. Next, each case
with low casewise weight is 
represented by a shade of grey overlaid on its
row. Note the resemblance with the rightmost
panel of Figure~\ref{fig:casecellmixed}.

\cite{Hirari:cellGraphics} also construct other 
graphical displays. \cite{Pfeiffer:cellsparsePCA} 
propose a cellwise robust PCA method that yields 
sparse loadings.

The casewise influence function (IF) is a 
fundamental tool in robust statistics 
\citep{hampel1986robust}. It measures an 
estimator’s sensitivity to a small mass of
contamination in different positions in the data 
space. \cite{alqallaf2009} proposed a cellwise 
version, and computed it for a multivariate 
estimator of location. The cellwise IF of 
cellPCA is derived in \cite{cellPCA}. It is 
bounded and redescending. 

\subsection{Regularized covariance matrices} 
\label{sec:cellRCov}

The covariance estimators discussed in 
Section~\ref{sec:cov} require that $n > d$, and 
most are computationally feasible in up to about 
20 dimensions. In high-dimensional settings,
nonrobust regularized covariance estimators have 
been constructed to avoid ill-conditioned or 
singular covariance estimates. \cite{cellRCOV} 
proposed a robust version called Cellwise 
Regularized Covariance (cellRCov) that works well 
in the presence of casewise and cellwise outliers.
It decomposes the covariance matrix into two 
components. Given a $d$-dimensional random 
variable $X$ with covariance matrix $\bSigma$, 
we can express it as the sum of a component 
lying in the $k$-dimensional subspace of 
$\mathbb{R}^{d}$ spanned by the first 
$k$ eigenvectors of $\bSigma$, and another in 
its orthogonal complement:
\begin{equation*}
   X=X^k + X^{\perp}\, .
\end{equation*}
For a given data set of size $n$, we first 
standardize its variables robustly, yielding 
the dataset $\bZ$ in $\mathbb{R}^{n \times d}$. 
The cellRCov estimate of $\bZ$ is then defined 
as 
\begin{equation*}
  \bhSigma(\bZ) := 
  \bhSigma_{\bZ^k}+\bhSigma_{\bZ^{\perp}}\;.
\end{equation*}
Here the first covariance estimate 
$\bhSigma_{\bZ^k}$ is obtained from cellPCA 
applied to $\bZ$, whereas 
$\bhSigma_{\bZ^{\perp}}$ is computed as a 
regularized weighted covariance matrix of 
the residuals $\bZ^{\imp}-\widehat{\bZ}$
of the imputed data. The resulting covariance 
matrix $\bhSigma(\bZ)$ is then transformed 
back to the unstandardized data.

Making use of cellRCov, a cellwise and casewise 
robust method for canonical correlation analysis 
was constructed also.

\subsection{Estimating a precision matrix} 
\label{sec:GLASSO}

In many applications it is the inverse 
of the covariance matrix, also known as 
the precision matrix, that is of major 
interest. 
In low dimensions, it is of course 
possible to estimate a covariance matrix 
and then to invert the estimate, but 
this is no longer the case in 
high dimensions because the estimated
covariance matrix is often singular. 
Rather than regularizing the covariance 
matrix, one can then estimate the 
precision matrix and regularize it
at the same time. 
This induces zeroes in the precision 
matrix, which correspond to partial 
correlations of zero which in turn 
correspond to conditional independence 
in the Gaussian graphical model. 
A popular estimator of a high-dimensional 
precision matrix is the graphical lasso 
of \cite{friedman2008sparse}. 
It requires the input of an estimate 
of the covariance matrix $\bhSigma$, 
for which the classical empirical 
covariance is often used. The precision 
matrix $\bTheta$ is then estimated as
\begin{equation} \label{eq:GLASSO}  
  \bhTheta  = 
  \argmin_{\bTheta \succcurlyeq 0} 
	\left(\tr\left(\bhSigma \bTheta\right) 
	- \log(\det(\bTheta)) + 
	\lambda \sum_{j\neq \ell} 
    |\bTheta_{j\ell}|\right)\;.
\end{equation}
A natural question is whether a robust 
version of $\bhSigma$ would result in a 
robust precision matrix estimator, and 
this is indeed the case. 
This idea has been studied 
empirically and theoretically by 
\cite{ollerer2015robust}, 
\cite{croux2016robust}, 
\cite{tarr2016}, \cite{loh2018high}, 
and \cite{katayama2018robust}, using 
different $\bhSigma$ based on robust 
pairwise correlations.
The $\bhSigma$ in these studies 
are typically of the form 
\begin{equation}
 \bhSigma_{j\ell} = \mbox{s}\!\left(
  \bX_{\bdot j}\right)\mbox{s}\!\left(
  \bX_{.\ell}\right) \mbox{r}\!\left(
  \bX_{\bdot j}, \bX_{.\ell}\right)
\end{equation}
where $\mbox{s}(\cdot)$ denotes a robust 
scale estimator, often the $Q_n$ estimator 
of \cite{rousseeuw1993alternatives}, and 
$\mbox{r}(\cdot)$ denotes a robust 
correlation estimator such as Spearman's 
rank correlation. 
Note that $\bhSigma$ needs to be positive 
semidefinite for~\eqref{eq:GLASSO} to 
work. When $\bhSigma$ is not positive 
semidefinite it is first modified, e.g.\
by adding a small multiple of the 
identity matrix. 
The resulting precision matrix 
estimators perform well under 
non-adversarial cellwise contamination, 
are fast to compute, and can be analyzed
theoretically due to the well-defined 
objective function~\eqref{eq:GLASSO}. 
But since they focus on pairwise 
correlations, their performance is 
likely to suffer under more adversarial 
cellwise outliers. 

\section{TENSOR DATA}
\label{sec:tensor}
Tensors extend two-way matrices into 
higher-dimensional arrays, and contain many 
cells, some of which may be corrupted. 
Casewise robust methods for tensor data 
discard or downweight complete cases 
(tensors), see e.g. \citep{Pravdova:RobTuck3,
Engelen:robParafac, Hubert:RParafac-SI,
Mayrhofer:RobCovMatrix}. This however leads to 
a large loss of efficiency when cases contain
only a few outlying cells.

Similarly to PCA, tensor decompositions aim to 
reduce the dimensionality of the data. Given a 
set of $n$ independent tensors 
$\mathcal{X}_i=\sbr{x_{i, p_1 \ldots p_L}} \in 
\mathbb{R}^{P_1 \times \cdots \times P_L}$ (for 
$i=1,\ldots,n)$, the objective of Multilinear 
Principal Component Analysis (MPCA) is to find 
a collection of  orthogonal matrices 
$\bV^{(\ell)} = \sbr{v_{p_\ell k_\ell}^{(\ell)}}
\in \mathbb{R}^{P_\ell \times K_\ell}$ of rank 
$K_\ell \leqslant P_\ell$, together with a center 
$\mathcal{C}=\sbr{c_{p_1 \ldots p_L}} \in 
\mathbb{R}^{P_1\times \cdots \times P_L} $ and 
core tensors 
$\mathcal{U}_i = [u_{i,k_1\ldots k_L}] \in 
\mathbb{R}^{K_1 \times \cdots \times K_L}$, such 
that the reconstructed tensors
$\widehat{\mathcal{X}}_i$ with entries
\begin{equation} 
\label{eq:MPCA}
 \widehat{x}_{i,\idx} = c_{\idx} + 
 \sum_{k_1=1}^{K_1} \sum_{k_2=1}^{K_2} \cdots 
 \sum_{k_L=1}^{K_L}  u_{i,k_1\ldots k_L} 
 v_{p_1k_1}^{(1)} \dots v_{p_Lk_L}^{(L)}
\end{equation}
are good approximations of the original tensors 
$\mathcal{X}_i$ \citep{Lu:MPCA}. The decomposition 
in~\eqref{eq:MPCA} is called a Tucker 
decomposition. 
When all $K_\ell$ are equal and the core tensors 
are required to be superdiagonal (which means that 
the only non-zero elements are those where all 
indices are identical), a CANDECOMP/PARAFAC (CP) 
decomposition is obtained 
\citep{BallardKolda:bookTensor, Bi:Tensors}.

\cite{Inoue:RMPCA} introduced robust adaptations
of MPCA that can either deal with casewise or 
cellwise outliers in the data, but not 
simultaneously. They either apply a bounded $\rho$ 
function to the Frobenius norm of the residual 
tensor $\mathcal{X}-\widehat{\mathcal{X}}$, or to
its entries. 

The robust multilinear PCA (ROMPCA) method of 
\cite{Hirari:ROMPCA} can 
cope with both types of outliers simultaneously, 
and allows for missing values in the tensors as 
well. It generalizes ideas from cellPCA to the 
tensor setting by combining two bounded $\rho$ 
functions that act at the casewise and at the 
cellwise level. For each case it also provides 
an imputed tensor 
$\mathcal{X}^{\imp}_i$ and a weight 
$w^x_i$ that is 0 when either 
$\|\mathcal{X}_i-\widehat{\mathcal{X}}_i\|_F$ 
is too large, or when its estimated core tensor 
deviates highly from the regular core tensors. 

Residual cellmaps can be drawn for the unfolded 
tensors, or for some slices. Consider for example
the Dog Walker dataset  which is a surveillance 
video of a man walking his dog, available at
\url{https://www.wisdom.weizmann.ac.il/~vision/SpaceTimeActions.html}\,.\linebreak
The video was
filmed using a static camera and contains 54 color 
frames. Each case is a frame with $72 \times 90$ 
pixels stored using the RGB (Red, Green, Blue) 
color model, and thus corresponds with a 3rd order 
tensor. The top row of Figure~\ref{fig:DogWalker}
contains some frames, and the second row shows the 
corresponding ROMPCA residual maps of the blue 
color. The first and last frames also show numerous 
deviating pixels in the background. They mostly 
result from tree leaves moving in the wind. The 
residual cellmap in the 3rd row is 
obtained by vectorizing each frame and stacking 
the vectors rowwise. Adjacent 
columns are aggregated, and their average 
color is displayed. The last panel depicts the 
cellmap as a 3D array. We clearly see that most of 
the outlying cells correspond to the man walking 
his dog. He is successfully flagged as outlying, 
as this is the only part that changes 
substantially in the successive frames.  

\begin{figure}[ht]
\centering
\includegraphics[width=1.0\textwidth]
  {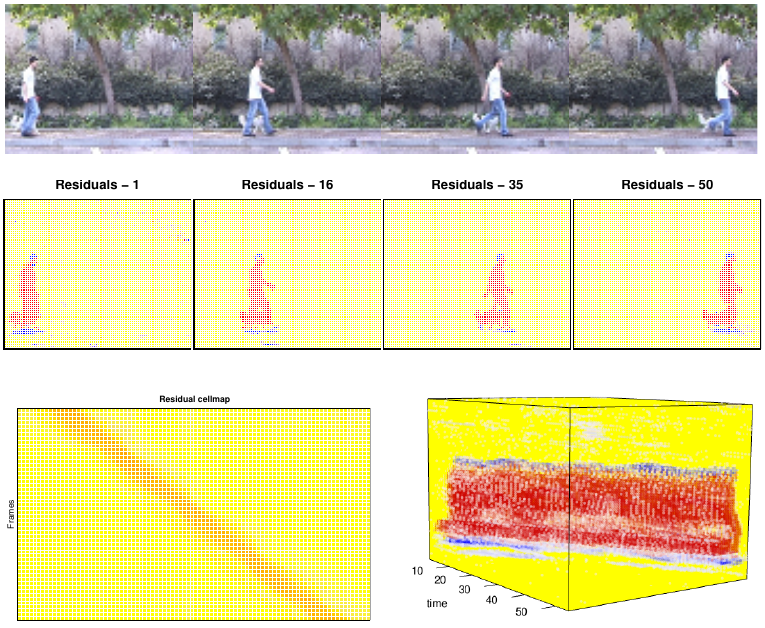}
\caption{Some frames of the Dog Walker data, and 
the corresponding ROMPCA cellmaps. The third row 
shows the residual cellmap of the unfolded video 
frames and a 3D cellmap.}
\label{fig:DogWalker}
\end{figure}

\cite{Hubert:MacroPARAFAC} introduced a casewise 
and cellwise robust method for CP decomposition. 
Their MacroPARAFAC method starts by applying 
FastDDC to the unfolded tensors, and then follows 
several principles underlying MacroPCA. It 
estimates loadings and the core vectors (the 
entries on the superdiagonal of the core tensor) 
iteratively by an alternating least squares 
PARAFAC algorithm for incomplete data, thereby 
updating the imputations for the cellwise outliers 
and missing values. 

The linear regression model in~\eqref{eq:reg} 
can be generalized to tensor-on-tensor (TOT) 
regression where both the predictor 
$\mathcal{X}_i \in \mathbb{R}^{P_1 \times \cdots 
\times P_L}$ and the response $\mathcal{Y}_i 
\in \mathbb{R}^{Q_1 \times \cdots \times Q_M}$
are tensors. The TOT regression model of 
\cite{lock:TOT} assumes that 
\begin{equation*}
 y_{i,q_1\cdots q_M} = \beta_{0,q_1 \cdots q_M} + 
 \sum_{p_1=1}^{P_1} \cdots \sum_{p_L=1}^{P_L}  
 x_{i,p_1 \cdots p_L} 
 \beta_{p_1 \cdots p_L q_1 \cdots q_M} + 
 \eps_{i,p_1 \cdots p_L q_1 \cdots q_M} 
\end{equation*}
where $\mathcal{B}_0 \in \mathbb{R}^{Q_1 \times 
\cdots \times Q_M}$ is the intercept, 
$\mathcal{B} \in \mathbb{R}^{P_1 \times \cdots 
\times P_L \times Q_1\times \cdots \times Q_M}$ 
is the slope tensor, and $\mathcal{E}_i \in 
\mathbb{R}^{Q_1 \times \cdots \times Q_M}$ is 
the error tensor. Moreover, it is assumed that 
$\mathcal{B}$ can be represented by a CP 
decomposition of rank $R$. \cite{Lee:RTOT} 
introduced a robust adaptation of TOT regression 
that can handle cellwise outliers in the response 
tensor. \cite{Hirari:ROTOT} proposed the ROTOT 
method that can also cope with casewise and 
cellwise outliers in the predictors. It starts 
by replacing the predictors with their ROMPCA 
imputed tensors, and then applies a weighted 
ridge M-regression to the cases whose ROMPCA 
weight $w_i^x$ is different from 0. 
Out-of-sample predictions are obtained as in 
Figure~\ref{fig:cellLTS}, based on the ROMPCA 
imputations of the predictors. 

\section{OTHER MODELS}
\label{sec:other}

Apart from the estimation of location and
covariance matrices, regression, PCA,
and tensor analysis covered in the preceding 
sections, the cellwise paradigm has made
inroads in many other statistical settings.
We will give a brief summary.

\vspace{1mm}
\begin{itemize}
\item \textbf{Detecting outlying
cells in other settings.}
The DDC method described in 
Section~\ref{sec:detect} operates on
numeric data, and has been extended to
some other data types. 
\cite{walach2020cellwise} used a 
similar approach for metabolomics
data by aggregating outlyingness over 
pairwise log-ratios. \cite{markatou2025}
extended DDC to contingency tables, in
the context of detecting side effects of
medical products.

\vspace{1mm}
\item \textbf{Cellwise robust 
discriminant analysis.} 
\cite{aerts2017cellwise} construct
cellwise robust versions of linear 
and quadratic discriminant analysis.
To this end they estimate a cellwise
robust precision matrix of the type of 
\cite{croux2016robust} for each class, 
and plug these into the assignment rule. 
More recently \cite{centofanti2026cellda} 
introduced a robust out-of-sample 
prediction rule, for dealing with new data 
that themselves contain outlying cells.

\vspace{1mm}
\item \textbf{Cellwise robust cluster
analysis.}
Several methods for clustering data with 
cellwise outliers have been proposed. 
The earliest was \cite{Farcomeni2014}, 
who used a mixture model with an 
additional $n \times \di$ parameter 
matrix of zeroes and ones, indicating 
which cells should be flagged and set 
to missing. 
The goal was to maximize the observed 
likelihood of the unflagged cells. 
The method was implemented in the 
\textsf{R} package \texttt{snipEM} 
\citep{Farcomeni2019} that is no longer 
maintained. 

\cite{garcia2021cluster} consider a
model in which clusters are linear,
i.e. they lie near lower-dimensional
subspaces. Each cluster is then fitted
by a cellwise robust PCA in the style 
of \cite{maronna2008robust}. 
They minimize the robust scales
of the coordinatewise residuals as in
least trimmed squares regression. 
The algorithm iterates the following
three steps: (i) updating the subspace 
parameters by robust PCA, (ii) updating 
the weights indicating which cells are
deemed outlying, and (iii) updating 
the group memberships as in the usual 
$k$-means clustering algorithm.

The cellGMM method of 
\cite{zaccaria2025cluster}
is aimed at finding ellipsoidal clusters
instead. The authors assume a Gaussian 
mixture model (GMM), in which they use 
the cellMCD method described in 
Subsection~\ref{sec:cov_methods} to fit 
centers and covariance matrices and to 
impute outlying cells. This is alternated 
with EM-type updates of the cluster 
memberships by constrained optimization of 
the cellGMM objective function.
\cite{zaccaria2025fuzzy} extend this work
to the cellFCLUST method for fuzzy
clustering.

\vspace{1mm}
\item \textbf{Time series.}
An application of cellwise robust methods
to autoregressive time series was described 
in \cite{challenges}.
\cite{Rodessa} developed a cellwise robust 
version of singular spectrum analysis 
called RODESSA, using a kind of dimension 
reduction.

\vspace{1mm}
\item \textbf{Functional data.}
In functional data analysis, a case
is not a row of a matrix but a function 
such as a curve. 
In that setting, cellwise outliers
correspond to local deviations (like
spikes) in a curve, rather than global
outlyingness of the entire curve.
In the taxonomy of \cite{Hubert2015SMA}
of functional outliers these are
called {\it isolated outliers}, and 
the authors constructed methods to detect 
them. \cite{garcia2021cluster} apply
cellwise robust methods to functional 
data.

\vspace{1mm}
\item \textbf{Compositional data.}
For regression with compositional 
covariates, \cite{vstefelova2021robust} 
propose a four-step procedure: 
(i) detecting outlying cells by DDC, 
(ii) imputing the outlying cells that 
are not in outlying cases, 
(iii) running  robust regression
on the imputed data, and 
(iv) performing inference based on multiple 
imputation.
\cite{Templ2026} made a methodological study
of the effect of cellwise contamination on 
compositional data.
\end{itemize}

\vspace{1mm}
\cite{agostinelli2016composite} developed
a cellwise approach for linear mixed models,
and \cite{Peng2025} studied cellwise
robust conformal inference.

\section{CONCLUSION AND OUTLOOK} 
\label{sec:conc}

Cellwise outliers form a new paradigm in
statistics and machine learning. We have 
seen that addressing them requires
abandoning several cherished types of
equivariance. But even though the problem
was tough, the community has stepped up
and progress has been made. Here we 
surveyed developments in the
detection of cellwise outliers as well as 
robust estimation of covariance matrices,
linear regression, dimension reduction,
tensor data, and other settings. 

Much work remains to be done on other 
frameworks, algorithms, and graphical 
displays to facilitate interpretation 
and communicating results. Here are some
topics for future research:
\begin{itemize}
\vspace{-2mm}
\item cellwise robust independent
  component analysis;
\item cellwise robust regression 
  for functional data;
\item detecting deviating data cells
  when the dataset is of mixed type,
  with both numerical and categorical
  variables;
\item cellwise robust versions of 
  chemometrics tools, such as partial 
  least squares and multivariate curve
  resolution.\\
\end{itemize}

\noindent{\bf Software availability.}
The code for several cellwise robust 
methods described here is available in 
the \textsf{R} package \texttt{cellWise} 
\citep{cellWise}. Other software is
listed in the text or referenced in
the cited papers.\\

\noindent{\bf Acknowledgment.}
We are grateful to Fabio Centofanti for
interesting discussions on this topic.

\spacingset{1.1}


\clearpage
\pagenumbering{arabic}
\appendix
\section*{Appendix}

\setcounter{equation}{0} 
\renewcommand{\theequation}
  {A.\arabic{equation}} 

\spacingset{1.12} 

Here we give some technical
formulations and proofs of results
mentioned in the text.

For estimating the location $\bmu$ of 
multivariate data, we can prove that any 
orthogonally equivariant estimator 
$\bhmu$ satisfies the natural sounding 
property
\begin{equation}
\begin{aligned}\label{eq:locEFP}
 \mbox{if all points of } \bX
 \mbox{ lie in a lower-dimensional}\\ 
 \mbox{ affine subspace, then } 
 \bhmu(\bX)
 \mbox{ lies in that subspace as well.}
\end{aligned}
\end{equation}
Next, one shows that all the points of $\bX$ 
can be moved to a hyperplane of choice by 
changing just one cell in each row, which can
be done in a way that no more than 
$\lceil n/d \rceil$ cells are changed in each 
column of $\bX$. The proofs, and those of 
the next results, can be found in the appendix 
of \cite{challenges}.

Analogously, we can show that any affine
equivariant covariance matrix $\bhSigma$
satisfies
\begin{equation}\label{eq:covEFP}
\begin{aligned}
 \mbox{if all points of } \bX
 \mbox{ lie in a lower-dimensional}\\ 
 \mbox{affine subspace, then } 
 \bhSigma(\bX)
 \mbox{ is singular.}
\end{aligned}
\end{equation}
Shifting a small number of cells per column
then yields Proposition~\ref{prop:impl}.

In the regression setting, one can prove that
the combination of regression and scale
equivariance implies:
\begin{equation}\label{eq:regEFP}
\begin{aligned}
 \mbox{if } y_i = 
 [1\,;\,\bx_i^{\btop}] \bgamma_0 
 \mbox{ for all $i=1,\ldots,n$ and } 
 \bX 
 \mbox{ is not}\\
 \mbox{in a lower-dimensional}
 \mbox{ subspace, then }
 \bhgamma = \bgamma_0\;.
\end{aligned}
\end{equation}
Replacing a small number of cells
then proves Proposition~\ref{prop:reg}.

For the breakdown value of a PCA method, 
we consider the following property:
\begin{equation}
\begin{aligned}\label{eq:pcaEFP}
 \mbox{if all points of } \bX
 \mbox{ lie in an} 
 \mbox{ affine subspace,}\\ 
 \mbox{then }
 \bhP(\bX)
 \mbox{ lies in that subspace as well.}
\end{aligned}
\end{equation}

\begin{proposition} 
The cellwise breakdown value of a 
PCA estimator $\bhP$ satisfying
condition~\eqref{eq:pcaEFP} at a dataset 
$\bX$ for which $\bhP(\bX)$ is not 
parallel to any coordinate axis, is 
bounded by
\begin{equation*}
  \varepsilon_n^*(\bhP, \bX) \leqslant
  \left\lceil \frac{n-1}{d} 
  \right\rceil/n\, \approx \frac{1}{d}\;.
\end{equation*}
\end{proposition}

\begin{proof}
Take any vector $\bv$ in the principal 
subspace $\bhP(\bX)$ and consider 
$\bP_0 = \bhP(\bX) - \bv$ which contains the 
origin. Construct $k$ basis vectors 
$\be_1,\ldots,\be_k$ of $\bP_0$.
Since $k < d$ the orthogonal complement of 
$\bP_0$ is nonempty, and it has
a basis $\be_{k+1},\ldots,\be_d$\,. Now 
consider the hyperplane $\bH$ spanned by 
$\be_2,\ldots,\be_d$\,. It satisfies
$\mbox{maxangle}\big(\bH, \bhP(\bX)\big) = 
\pi/2$ because $\be_1$ is orthogonal
to all vectors in $\bH$. We are 
then again in the situation 
of Proposition~\ref{prop:loc}, where 
in each row of $\bX$ we can move one cell
to obtain a point on $\bH$. 
Condition~\eqref{eq:pcaEFP} then implies 
that $\bhP(\bX^m) \subset \bH$, so 
$\mbox{maxangle}\big(\bhP(\bX^m),\bhP(\bX)
\big) = \pi/2$.
\end{proof}

We now show that any orthogonally 
equivariant PCA method satisfies 
property~\eqref{eq:pcaEFP}. 

\begin{proof}
Without loss of generality we can consider 
subspaces through the origin (by shifting). 
Suppose $\bhP$ is an orthogonally equivariant 
estimator, i.e. 
$\bhP(\bX \bR^{\top}) = \bR\,\bhP(\bX)$
for any orthogonal matrix $\bR$ and any 
dataset $\bX$.

Due to the orthogonal equivariance of $\bhP$, 
we can assume $\bX$ lies in the subspace 
spanned by the first $k$ canonical basis 
vectors $\be_1, \ldots, \be_k$. If this is not 
the case, we can transform $\bX$ by orthogonal 
transformation so that $\bX$ lies in
$\sspan(\be_1, \ldots, \be_k)$. This means 
that the last $d-k$ columns of $\bX$  are all 
zeroes. Now consider the diagonal matrix 
$\bD = \begin{bmatrix}
\bI_k & 
\bzero_{k\times (d-k)}\\
\bzero_{(d-k)\times k} & 
-\bI_{(d-k)}\\
\end{bmatrix}$. This orthogonal matrix 
$\bD$ reflects the last $d-k$ columns of 
$\bX$. Since these are zero, we have 
$\bX \bD^{\top} = \bX$. Now, because 
$\bhP$ is orthogonally equivariant, we have 
$\bhP(\bX) = \bhP(\bX \bD^{\top}) = 
\bD \bhP(\bX)$. This implies that the last 
$n-k$ components of $\bhP(\bX)$ are zero, 
hence \mbox{$\bhP(\bX) \subseteq 
\sspan(\be_1,\ldots, \be_k)$}.
\end{proof}

Together, these results prove 
Proposition~\ref{prop:PCA}.

\end{document}